\documentclass[12pt]{llncs}
\usepackage[margin=1in]{geometry} 
\usepackage{amsmath,amssymb}
\usepackage{marvosym}
\usepackage{subfigure}
\usepackage{color}
\usepackage{algorithm}
\usepackage{algorithmic}
\usepackage{calc}
\usepackage{tikz}
\usepackage{frame}

\newcommand{\BCLIQUE}{\textsc{BClique}}
\newcommand{\UCLIQUE}{\textsc{UClique}}

\title{Deterministic graph connectivity in the broadcast congested clique}
\author{Pedro Montealegre\thanks{Research supported by Conicyt-Becas Chile-72130083.}  \and Ioan Todinca}

\institute{
Univ. Orl\'{e}ans, INSA Centre Val de Loire, LIFO EA 4022,  Orl{\'e}ans , France
\\ \texttt{(pedro.montealegre $\mid$ ioan.todinca)@univ-orleans.fr}
}

\begin{document}
\maketitle
\begin{abstract}
We present deterministic constant-round protocols for the graph connectivity problem in the model where each of the $n$ nodes of a graph receives a row of the adjacency matrix, and  broadcasts a single sublinear size message to all other nodes. Communication rounds are synchronous. This model is sometimes called the broadcast congested clique. Specifically, we exhibit a deterministic protocol that computes the connected components of the input graph in $\lceil 1/\epsilon \rceil$ rounds, each  player communicating $\mathcal{O}(n^{\epsilon} \cdot \log n)$ bits per round, with $0 < \epsilon \leq 1$. 

We also provide a deterministic one-round protocol for connectivity, in the model when each node receives as input the graph induced by the nodes at distance at most $r>0$, and communicates $\mathcal{O}(n^{1/r} \cdot \log n)$ bits.  This result is based on a $d$-pruning protocol, which consists in successively removing nodes of degree at most $d$ until obtaining a graph with minimum degree larger than $d$. Our technical novelty is the introduction of deterministic sparse linear sketches: a linear compression function that permits to recover sparse Boolean vectors deterministically. 

\end{abstract}

\keywords{broadcast congested clique; graph connectivity; spanning forest;  deterministic protocol}

\section{Introduction}
This paper proposes the first (to our knowledge) constant-round deterministic protocol that computes a  spanning forest\footnote{In this paper a spanning forest designs a maximal one.} in the \emph{broadcast congested clique} model with {\it sub-linear} message size. 

The \emph{congested clique}  model is a message-passing model of distributed computation where $n$ nodes communicate with each other in synchronous rounds over a {\emph{complete network}}~\cite{AGM12,AGM12b,BKM+15,drucker12}.
The joint input of the $n$ nodes is an undirected  graph $G$ on the same set of nodes, with node $v$ receiving the list of its neighbors in $G$. Nodes have pairwise distinct identities, which are numbers upper bounded by some polynomial in $n$. Moreover, all nodes know this upper bound. At the end of a protocol, each node must produce the same output. In the  {\emph{broadcast congested clique}}  (denoted $\BCLIQUE[f(n)]$) each node broadcasts, in each round, a single  $\mathcal{O}(f(n))$-sized message\footnote{Some authors consider that in the \emph{congested clique} model only $\mathcal{O}(\log n)$-size messages are allowed. In our case the {\it congestion} refers to sublinear capacity communication channels.} along each of its $n-1$ communication links, while in the {\emph{unicast congested clique}} (denoted $\UCLIQUE[f(n)]$) nodes are allowed to send  {\emph{different}}  $\mathcal{O}(f(n))$-size messages through different links.

There exists a basic $\mathcal{O}(\log n)$-round connectivity protocol in the $\BCLIQUE[\log n]$ model. Each node successively sends a neighbor outside its current cluster (initially the clusters consist in single nodes), and at each round the adjacent clusters are merged. The protocol can be adapted to compute a maximum weight spanning forest. Despite its simplicity, only few results exist improving the basic connectivity protocol. 

On one hand, Lotker {\it et al.} \cite{doi:10.1137/S0097539704441848} proposed the first improvement in the unicast congested clique, namely a $\mathcal{O}(\log\log n)$-round protocol for maximum weight spanning forest in the $\UCLIQUE[\log n]$ model, which also runs in $\mathcal{O}(1/\epsilon)$ rounds in  $\UCLIQUE[n^{\epsilon}]$, for any constant $\epsilon >0$. Later, Hegeman {\it et al.} \cite{Hegeman:2015:TOB:2767386.2767434} improved the number of rounds  to $\mathcal{O}(\log \log \log n)$, also in $\UCLIQUE[\log n]$, but in a randomized protocol, which outputs a maximum weight spanning forest with high probability.  On the other hand, in the broadcast congested clique, a one-round randomized protocol for graph connectivity in $\BCLIQUE[\log^3 n]$ was proposed by Ahn {\it et al.} \cite{AGM12}, using an elegant techinque called {\it linear sketches}. 

\medskip
\noindent\textbf{Our results.}
First, we observe that the basic connectivity protocol can be used to produce, for each $0 < \epsilon \leq 1$  a $\mathcal{O}(1/\epsilon)$-round protocol that computes a spanning forest in the $\BCLIQUE[n^{\epsilon} \cdot \log n]$ model. 

We also propose a one-round connectivity protocol, this time the $\BCLIQUE_r[n^{1/r}\log n]$ model. This corresponds to the broadcast congested clique model, when each node receives as input the graph induced by the nodes at distance at most $r$ (one may think of this model as the one where vertices are allowed to perform $r$ unrestricted {\it local} communications, and then some congested {\it global} communications).

Our main tool is a protocol to {\it prune} a graph. The $d$-{\it pruning} of graph $G$, for $d>0$, consists of finding a sequence of nodes $(x_1,\dots, x_p)$, together with their incident edges, such that each vertex $x_i$ has degree at most $d$ in the graph obtained by removing vertices $\{x_1,\dots, x_{i-1}\}$ from $G$, and the graph obtained by removing the whole sequence has minimum degree strictly larger than $d$.

Becker \textit{et al.}~\cite{BKM+15} provide a one-round deterministic protocol for $d$-pruning a graph in $\BCLIQUE[d^2 \cdot \log n]$. They actually work on the reconstruction of d-degenerate graphs, i.e., graphs for which the pruning sequence is the whole vertex set. 

We give a new protocol for $d$-pruning, which also runs in one round, but improves the message size to $\mathcal{O}(d \cdot \log n)$. Our protocol is optimal since the number of $d$-degenerate graphs is $2^{\Omega(nd\log n)}$.  This new protocol is inspired by the {\it fingerprint} technique, which uses the Schwartz-Zippel Lemma for equality testing in randomized communication complexity \cite{KuNi97}. The protocol is based on a derandomized version of fingerprints to compute {\it sparse linear sketches}: a linear compression function that permits to recover $d$-sparse Boolean vectors deterministically. This technique has its own interest and we believe that can be used in other applications.

\section{Pruning and connectivity}

\begin{theorem}
Let $0 < \epsilon \leq 1$. There is a $\lceil 1/\epsilon \rceil$-round deterministic protocol that computes the connected components of the input graph in the $\BCLIQUE[n^{\epsilon} \cdot \log n]$ model. The protocol returns a spanning forest of the input graph.
\end{theorem}
\noindent {\it Proof sketch.}
Let $G=(V,E)$ the input graph. The protocol is very similar to the basic connectivity protocol described in the introduction. First, denote $\hat{V}$ the set of {\it supernodes}, which initially are the $n$ singletons $\{\{u\} | u \in V\}$. At each round,  every node sends $n^{\epsilon}$ arbitrary neighbors in pairwise different supernodes (if it sees less than $n^{\epsilon}$ supernodes, then it sends one neighbor for each of them). At the end of each round adjacent supernodes are merged.

A supernode of round $t>1$ is called {\it active} if it contains at least $n^{\epsilon}$ supernodes of round $t-1$, and otherwise is called {\it inactive}. Note that at any round $t>0$, an inactive supernode corresponds to a connected component of $G$. The protocol finishes when every supernode is inactive. 
Let $n_t$ be the number of active supernodes at round $t \geq 0$. Since the number of active nodes at round $t+1$ at most $n_t / n^{\epsilon}$, the protocol stops in at most $\lceil 1/\epsilon \rceil$ rounds. 
\qed



\medskip

We remark that when $\epsilon \leq 1/(\log n)$ our protocol matches de basic one. When $\epsilon = 1/2$, we obtain a sublinear, two-round protocol for connectivity. To our knowledge, there are no nontrivial lower bounds for deterministic protocols solving graph connectivity in the broadcast congested clique, even restricted to one-round protocols. However, even in the powerful unicast congested clique, there are no known sublinear one-round deterministic protocols. In the following, we propose a one-round deterministic protocol for graph connectivity in the $\BCLIQUE_r$ model.

\begin{theorem}\label{theo:prun}
There exists a deterministic protocol in the $\BCLIQUE[d \cdot \log n]$ model that computes a $d$-{pruning} in one round.
\end{theorem}
Let $d>0$. We say that a vector is $d$-sparse if it has at most $d$ nonzero coordinates. We show that there exists a linear function compressing integer vectors, and such that if the compressed vector is Boolean and $d$-sparse, the vector can be recovered. We emphasize that the recovery property does not work over all $d$-sparse integer vectors, but only on  Boolean ones. In the following $\mathbb{F}_p$ denotes the field of integers {\it modulo} $p$, where $p>0$ is prime.

\begin{lemma}\label{lem:prun}
Let $n,d >0$. There exists a function $f: \mathbb{Z}^n \rightarrow \mathbb{F}_p$, for some prime number  $p = 2^{\mathcal{O}(d \cdot \log n)}$, such that: (1) $f$ is linear, and (2) $f$ is injective when restricted to $d$-sparse Boolean inputs.
\end{lemma}
\noindent {\it Proof sketch.}
Let $\mathcal{B} = \{ b \in \{0,1\}^n : \sum_{i=1}^n b_i \leq d\}$ be the family of $d$-sparse  Boolean vectors of dimension $n$, and $\mathcal{T} = \{t \in \{-1, 0, 1\}^n : \exists \text{~distinct~} b, b' \in \mathcal{B},  t= b - b'\}$. Call $p = p(n,d)$ the smallest prime number greater than $(1+n)^{2d}\cdot n$. Let $P(\mathcal{T})$ be the family of polynomials over the field $\mathbb{F}_p$ associating to each $t \in \mathcal{T}$ the polynomial $P(t,X) = \sum_{i=1}^n t_i X^{i-1}$ (values are taken modulo $p$). Let $\overline{x} = \overline{x}(n,d)$ be the minimum integer in $\mathbb{F}_p$ which is not a root of any polynomial  in $P(\mathcal{T})$; $\overline{x}$ exists because there are at most  $|\mathcal{T}|$ polynomials in $P(\mathcal{T})$,  each polynomial has at most $n$ roots in $\mathbb{F}_p$, and $p > |\mathcal{T}| \cdot n$. We define then for each $v \in \mathbb{Z}^n$ the function $f(v) = P(v,\overline{x})$. Clearly $f$ is linear, and by definition of $\overline{x}$, for any distinct $b, b' \in  \mathcal{B}$, we have $f(b) = P(b, \overline{x})\neq P(b', \overline{x}) = f(b')$. 
\qed
\medskip

\noindent {\it Proof sketch of Theorem \ref{theo:prun}.}  In the $d$-pruning protocol, each player $i$ sends the message $M_i = (M_i^1, M_i^2) = (d_i, f(a_i))$, where $d_i$ is its degree, $a_i$ is the row of the adjacency matrix corresponding to node $i$, and $f$ is the function of Lemma \ref{lem:prun}. The number of communicated bits is 
$\mathcal{O}(d \log n)$. Call  $M(G) = (M_1, \dots, M_n)$ the {\it messages vector} of $G$. The nodes use $M(G)$ to prune the graph as follows. First they look for a node $k$ such that $M^1_k\leq d$. Using the  injectivity property of $f$ they can obtain the neighborhood of $k$ in $G$. Then, using the linearity of $f$ they compute the messages vector of $G - \{k\}$ updating $M_j = ( M_j^1 - 1, M_j^2 - f(e_k))$ for each $j$ neighbor of $k$, where $e_k$ is the Boolean vector having a unique one in the $k$-th coordinate. The process is reiterated (with no extra communication) until there is no more node of degree at most $d$.
\qed



We use our $d$-pruning protocol to produce a one-round connectivity protocol, this time in the $\BCLIQUE_r[n^{1/r}\cdot \log n]$ model. Our protocol is based in the following propositions respectively found in \cite{Awerbuch}, and \cite{drucker12}. Let $G = (V,E)$ be a graph, $r>0$, and $\sigma$ some total ordering of $E$. For each cycle of length at most $2r$ pick the maximum edge of the cycle according to $\sigma$. Call $\tilde{E}$ the set of picked edges, and call $\tilde{G} = (V, E-\tilde{E})$. Note that $\tilde{G}$ has no cycles of length at most $2r$.

\begin{proposition}[\cite{Awerbuch}]
$G$ is connected if and only if $\tilde{G}$ is connected, and any spanning forest of $\tilde{G}$ is a spanning forest of $G$.
\end{proposition}

The following proposition shows that $\tilde{G}$ has {\it sublinear} degeneracy. Recall that the girth of a graph is the length of a shortest cycle.

\begin{proposition}[\cite{drucker12}]\label{prop:ex_deg}
Let $G$ be a graph of girth at least $2r$. Then $G$ is $\mathcal{O}(n^{1/r})$-degenerate.
\end{proposition}

\begin{theorem}
There is a one-round deterministic protocol that computes the connected components of the input graph in the $\BCLIQUE_r[n^{1/r}\cdot \log n]$ model. The protocol returns a spanning forest of the input graph.
\end{theorem}

\begin{proof} 

Choose an ordering of the edges of $G$, for example if we denote the edges $e= (u,v)$ with $u<v$, then $(u_1, v_1) < (u_2, v_2)$ if either $u_1 < u_2$ or $u_1 = u_2$ and $v_1 < v_2$. 

In the protocol, a node $v$ looks for all cycles of length at most $2r$ in $G$ that contain it. Notice that nodes do this without any communication since they see all neighbors at distance at most $r$. For each such cycle, $v$ picks the maximum edge according to the edge ordering, obtaining the row of the adjacency matrix of $\tilde{G}$ corresponding to $v$. Then each node can simulate the $s$-pruning protocol when the input graph is $\tilde{G}$, where $s = \mathcal{O}(n^{1/r})$ is the degeneracy of graphs of girth at least $2r$, obtained from Proposition \ref{prop:ex_deg}. Each node then reconstructs $\tilde{G}$ and computes a spanning forest of $\tilde{G}$.
\end{proof}

\section{Discussion}




We have shown that, in the $\BCLIQUE[n^{\epsilon} \cdot \log n]$ model, $\mathcal{O}(1/\epsilon)$ rounds are enough to decide connectivity deterministically. If nodes see the graph induced by the nodes at distance at most $r$, then connectivity can be decided by a deterministic one-round protocol, with messages of size $\mathcal{O}(n^{1/r}\cdot\log n)$.

These results rely on a deterministic protocol for $d$-pruning, which in particular reconstructs $d$-degenerate graphs. We believe that this protocol, and the sketch function it is based on, might be of interest for further applications.

Like Ahn {\it et al.}'s results~\cite{AGM12}, our connectivity protocol can be transformed into protocols to detect if the input graph is bipartite. Moreover, the $d$-pruning protocol can be implemented as a $\mathcal{O}(nd\log n)$-space protocol in the {\it  dynamic graph streaming} model.

To wrap up, let us recall that the existence of a one-round deterministic protocol for connectivity, in the broadcast congested clique with sublinear message size, is still open. We might as well ask whether the problem could be solved by randomized protocols using only \emph{private coins}.

\section{Acknowledgments}
We thank F. Becker and I. Rapaport for fruitful discussions on the subject, and the anonymous referees for useful suggestions. P. Montealegre would like to thank the support of Conicyt Becas Chile-72130083.


\bibliographystyle{abbrv}

\end{document}